\documentclass[copyright]{eptcs}

\usepackage{ou}
\usepackage{amsthm}
\usepackage{stmaryrd}

\makeatletter
\@ifundefined{lhs2tex.lhs2tex.sty.read}%
  {\@namedef{lhs2tex.lhs2tex.sty.read}{}%
   \newcommand\SkipToFmtEnd{}%
   \newcommand\EndFmtInput{}%
   \long\def\SkipToFmtEnd#1\EndFmtInput{}%
  }\SkipToFmtEnd

\newcommand\ReadOnlyOnce[1]{\@ifundefined{#1}{\@namedef{#1}{}}\SkipToFmtEnd}
\usepackage{amstext}
\usepackage{amssymb}
\usepackage{stmaryrd}
\DeclareFontFamily{OT1}{cmtex}{}
\DeclareFontShape{OT1}{cmtex}{m}{n}
  {<5><6><7><8>cmtex8
   <9>cmtex9
   <10><10.95><12><14.4><17.28><20.74><24.88>cmtex10}{}
\DeclareFontShape{OT1}{cmtex}{m}{it}
  {<-> ssub * cmtt/m/it}{}

\DeclareFontShape{OT1}{cmtt}{bx}{n}
  {<5><6><7><8>cmtt8
   <9>cmbtt9
   <10><10.95><12><14.4><17.28><20.74><24.88>cmbtt10}{}
\DeclareFontShape{OT1}{cmtex}{bx}{n}
  {<-> ssub * cmtt/bx/n}{}

\newcommand{\Conid}[1]{\mathit{#1}}
\newcommand{\Varid}[1]{\mathit{#1}}
\newcommand{\anonymous}{\kern0.06em \vbox{\hrule\@width.5em}}

\renewcommand{\leq}{\leqslant}

\usepackage{polytable}

\@ifundefined{mathindent}%
  {\newdimen\mathindent\mathindent\leftmargini}%
  {}%

\def\resethooks{%
  \global\let\SaveRestoreHook\empty
  \global\let\ColumnHook\empty}
\newcommand*{\savecolumns}[1][default]%
  {\g@addto@macro\SaveRestoreHook{\savecolumns[#1]}}
\newcommand*{\restorecolumns}[1][default]%
  {\g@addto@macro\SaveRestoreHook{\restorecolumns[#1]}}
\newcommand*{\aligncolumn}[2]%
  {\g@addto@macro\ColumnHook{\column{#1}{#2}}}

\resethooks

\newcommand{\onelinecommentchars}{\quad-{}- }
\newcommand{\commentbeginchars}{\enskip\{-}
\newcommand{\commentendchars}{-\}\enskip}

\newcommand{\visiblecomments}{%
  \let\onelinecomment=\onelinecommentchars
  \let\commentbegin=\commentbeginchars
  \let\commentend=\commentendchars}

\newcommand{\invisiblecomments}{%
  \let\onelinecomment=\empty
  \let\commentbegin=\empty
  \let\commentend=\empty}

\visiblecomments

\newlength{\blanklineskip}
\newcommand{\hsindent}[1]{\quad}
\let\hspre\empty
\let\hspost\empty

\EndFmtInput
\makeatother
\ReadOnlyOnce{polycode.fmt}%
\makeatletter

\newcommand{\hsnewpar}[1]%
  {{\parskip=0pt\parindent=0pt\par\vskip #1\noindent}}

\newcommand{\hscodestyle}{}

\newcommand{\sethscode}[1]%
  {\expandafter\let\expandafter\hscode\csname #1\endcsname
   \expandafter\let\expandafter\endhscode\csname end#1\endcsname}

  {\par\noindent
   \advance\leftskip\mathindent
   \hscodestyle
   \let\\=\@normalcr
   \let\hspre\(\let\hspost\)%
   \pboxed}%
  {\endpboxed\)%
   \par\noindent
   \ignorespacesafterend}

  {\hsnewpar\abovedisplayskip
   \advance\leftskip\mathindent
   \hscodestyle
   \let\hspre\(\let\hspost\)%
   \pboxed}%
  {\endpboxed%
   \hsnewpar\belowdisplayskip
   \ignorespacesafterend}

  {\hsnewpar\abovedisplayskip
   \advance\leftskip\mathindent
   \hscodestyle
   \let\\=\@normalcr
   \(\pboxed}%
  {\endpboxed\)%
   \hsnewpar\belowdisplayskip
   \ignorespacesafterend}

\newcommand{\plainhs}{\sethscode{plainhscode}}

\plainhs

  {\hsnewpar\abovedisplayskip
   \advance\leftskip\mathindent
   \hscodestyle
   \let\\=\@normalcr
   \(\parray}%
  {\endparray\)%
   \hsnewpar\belowdisplayskip
   \ignorespacesafterend}

  {\parray}{\endparray}

  {\(\parray}{\endparray\)}

\def\codeframewidth{\arrayrulewidth}
\RequirePackage{calc}

  {\parskip=\abovedisplayskip\par\noindent
   \hscodestyle
   \arrayrulewidth=\codeframewidth
   \tabular{@{}|p{\linewidth-2\arraycolsep-2\arrayrulewidth-2pt}|@{}}%
   \hline\framedhslinecorrect\\{-1.5ex}%
   \let\endoflinesave=\\
   \let\\=\@normalcr
   \(\pboxed}%
  {\endpboxed\)%
   \framedhslinecorrect\endoflinesave{.5ex}\hline
   \endtabular
   \parskip=\belowdisplayskip\par\noindent
   \ignorespacesafterend}

\newcommand{\framedhslinecorrect}[2]%
  {#1[#2]}

  {\(\def\column##1##2{}%
   \let\>\undefined\let\<\undefined\let\\\undefined
   \newcommand\>[1][]{}\newcommand\<[1][]{}\newcommand\\[1][]{}%
   \def\fromto##1##2##3{##3}%
   }{\) }%

  {\let\orighscode=\hscode
   \let\origendhscode=\endhscode
   \def\endhscode{\def\hscode{\endgroup\def\@currenvir{hscode}\\}\begingroup}
   \orighscode\def\hscode{\endgroup\def\@currenvir{hscode}}}%
  {\origendhscode
   \global\let\hscode=\orighscode
   \global\let\endhscode=\origendhscode}%

\makeatother
\EndFmtInput
\finaltrue

\newtheorem{thm}{Theorem}[section]
\newtheorem{lemma}[thm]{Lemma}
\newtheorem{definition}[thm]{Definition}
\newtheorem{corollary}[thm]{Corollary}

\newcolumntype{P}{@{}>{\let\pcr=\\\begin{center}\let\\=\pcr\(}p{\linewidth}<{\)\end{center}}@{}}

\title{Properties of Exercise Strategies}

\providecommand{\event}{IWS 2010}

\author{
Alex Gerdes \qquad Bastiaan Heeren
\institute{School of Computer Science\\
Open Universiteit Nederland\\
Heerlen, The Netherlands}
\email{age@ou.nl \qquad \qquad bhr@ou.nl}
\and
Johan Jeuring
\institute{Department of Information and Computing Sciences\\
Utrecht Univeristy\\
Utrecht, The Netherlands}
\email{johanj@cs.uu.nl}
}

\def\titlerunning{Properties of Exercise Strategies}
\def\authorrunning{A. Gerdes, B. Heeren \& J. Jeuring}

\begin{document}

\maketitle

\begin{abstract}
  Mathematical learning environments give domain-specific and immediate feedback
  to students solving a mathematical exercise. Based on a language for 
  specifying strategies, we have developed a feedback framework that 
  automatically calculates semantically rich feedback. We offer this feedback 
  functionality to mathematical learning environments via a set of web services.
  Feedback is only effective when it is precise and to the point. The tests we
  have performed give some confidence about the correctness of our feedback
  services. To increase confidence in our services, we explicitly specify the
  properties our feedback services should satisfy, and, if possible, prove them
  correct. For this, we give a formal description of the concepts used in our
  feedback framework services. The formalisation allows us to reason about these
  concepts, and to state a number of desired properties of the concepts. Our 
  feedback services use exercise descriptions for their instances on domains such 
  as logic, algebra, and linear algebra. We formulate requirements these domain
  descriptions should satisfy for the feedback services to react as expected. 
\end{abstract}

\section{Introduction}
\label{sec:intro}
We use strategies to calculate \emph{semantically rich feedback} for students
solving a mathematical exercise~\cite{heeren-10}. For example, we can calculate
a hint, or show a complete derivation for a number of mathematical domains, such
as propositional logic, linear algebra, and arithmetic. A strategy captures
expert knowledge about how to solve a particular problem. It describes which
steps a student can take to solve an exercise, and in what order. When a student
solves an exercise stepwise, we can check whether or not a step follows the
strategy.

We have developed an embedded domain-specific strategy language in which we can
specify strategies, and designed a feedback framework on top of it. This
framework is used to give detailed feedback to interactive mathematical
environments such as ActiveMath~\cite{activemath}, the Freudenthal Digital
Mathematics Environment~\cite{dwo}, and our own tool for rewriting logic
expressions~\cite{lodder}. A set of feedback services~\cite{gerdes-08} gives
mathematical environments access to our feedback functionality.

Feedback tops the list of factors leading to good learning~\cite{biggstang}, but
it is only effective when it is precise and to the point. Feedback messages
should not mislead a student practicing an exercise.  Therefore, we want to
ensure that the given feedback is to the point and relevant. The software we
have developed is augmented with numerous (unit) tests to test correct
behaviour. However, a formal definition of the concepts used, such as the
strategy language, exercises, and the services, contributes to the goal of
delivering proper feedback. In addition, we give a formulation, and if possible
a proof, of the properties they satisfy. The formalisation makes the concepts
precise, which enables us to reason about them. The feedback services use
exercise descriptions for their instances on domains such as logic, algebra, and
linear algebra. The exercise descriptions contain the strategy for solving a
particular exercise, a predicate that determines whether or not an exercise is
solved, etc. The desired properties of our feedback services lead to
requirements for the exercise descriptions. We will formulate these
requirements, and show how our formalisation helps in verifying that the
requirements are satisfied for a particular exercise.

This paper makes the following contributions:
\begin{itemize}
\item We give a formal definition of the concepts we use in our feedback 
framework.
\item We formulate the properties we want our feedback services to satisfy.
\item We formulate the requirements that an exercise description should fulfil,
  and we show these requirements are satisfied by an example exercise
  description.
\end{itemize}

This paper is organised as follows. Section~\ref{sec:framework} introduces the
fundamental concepts of our framework, such as rewrite rules and strategies.
Section~\ref{sec:services} lists the services we offer, and formulates the
properties that we want our services to satisfy. Section~\ref{sec:properties}
states a number of exercise-specific properties that an exercise description
should have. Section~\ref{sec:related} gives related work, and
Section~\ref{sec:conclusions} concludes.

\section{Strategy Language} 
\label{sec:framework} 
This section gives a formal definition of our rewrite-strategy language. We
start the introduction of our rewrite-strategy language for specifying exercises
with an example. Consider the problem of rewriting the expression \ensuremath{{({\Varid{a}^{\mathrm{3}}}\cdot{\Varid{a}^{\mathrm{4}}})^{\mathrm{2}}}} as a power of \ensuremath{\Varid{a}}, using the following three
rewrite rules:

\noindent\begin{hscode}\SaveRestoreHook
\column{B}{@{}>{\hspre}l<{\hspost}@{}}%
\column{3}{@{}>{\hspre}l<{\hspost}@{}}%
\column{30}{@{}>{\hspre}c<{\hspost}@{}}%
\column{30E}{@{}l@{}}%
\column{33}{@{}>{\hspre}l<{\hspost}@{}}%
\column{60}{@{}>{\hspre}l<{\hspost}@{}}%
\column{E}{@{}>{\hspre}l<{\hspost}@{}}%
\>[3]{}{\Varid{a}^{\Varid{x}}}\cdot{\Varid{a}^{\Varid{y}}}{}\<[30]%
\>[30]{}\mathrel{=}{}\<[30E]%
\>[33]{}{\Varid{a}^{\Varid{x}\mathbin{+}\Varid{y}}}\;{}\<[60]%
\>[60]{}\hspace{3cm}\;[\mskip1.5mu \textsc{AddExp}\mskip1.5mu]{}\<[E]%
\\
\>[3]{}{({\Varid{a}^{\Varid{x}}})^{\Varid{y}}}{}\<[30]%
\>[30]{}\mathrel{=}{}\<[30E]%
\>[33]{}{\Varid{a}^{\Varid{x}\cdot\Varid{y}}}\;{}\<[60]%
\>[60]{}\hspace{3cm}\;[\mskip1.5mu \textsc{MulExp}\mskip1.5mu]{}\<[E]%
\\
\>[3]{}{(\Varid{a}\cdot\Varid{b})^{\Varid{x}}}{}\<[30]%
\>[30]{}\mathrel{=}{}\<[30E]%
\>[33]{}{\Varid{a}^{\Varid{x}}}\cdot{\Varid{b}^{\Varid{x}}}\;{}\<[60]%
\>[60]{}\hspace{3cm}\;[\mskip1.5mu \textsc{DistExp}\mskip1.5mu]{}\<[E]%
\ColumnHook
\end{hscode}\resethooks
A possible strategy to solve the given power expression is to apply these rules
bottom-up, until none of the rules can be applied anymore. Applying this
strategy to the given expression results in the following derivation: \ensuremath{{({\Varid{a}^{\mathrm{3}}}\cdot{\Varid{a}^{\mathrm{4}}})^{\mathrm{2}}}\;{\; \stackrel{\textsc{AddExp}}{\Longrightarrow} \;}\;{({\Varid{a}^{\mathrm{7}}})^{\mathrm{2}}}\;{\; \stackrel{\textsc{MulExp}}{\Longrightarrow} \;}\;{\Varid{a}^{\mathrm{14}}}}. A strategy captures expert knowledge, in this case how to
solve an exercise involving powers. In addition to a strategy, an exercise
description consists amongst others of the rules that can be applied to a term,
the form of the term that is shown to the user (the exercise from the viewpoint
of the user), a predicate for determining whether or not an exercise is solved.

In the remainder of this section we give a formal definition of our strategy
language and related concepts. 

\subsection{Strategy Language Definition}

\begin{definition}
\label{def:strategy}
Let \ensuremath{\mu} be the fixed-point combinator \ensuremath{\mu\;\Varid{f}\mathrel{=}\Varid{f}\;(\mu\;\Varid{f})}. A strategy is
an element of the language of the following grammar:
\begin{hscode}\SaveRestoreHook
\column{B}{@{}>{\hspre}l<{\hspost}@{}}%
\column{3}{@{}>{\hspre}l<{\hspost}@{}}%
\column{6}{@{}>{\hspre}c<{\hspost}@{}}%
\column{6E}{@{}l@{}}%
\column{11}{@{}>{\hspre}l<{\hspost}@{}}%
\column{14}{@{}>{\hspre}l<{\hspost}@{}}%
\column{18}{@{}>{\hspre}c<{\hspost}@{}}%
\column{18E}{@{}l@{}}%
\column{26}{@{}>{\hspre}l<{\hspost}@{}}%
\column{30}{@{}>{\hspre}l<{\hspost}@{}}%
\column{39}{@{}>{\hspre}l<{\hspost}@{}}%
\column{43}{@{}>{\hspre}l<{\hspost}@{}}%
\column{52}{@{}>{\hspre}l<{\hspost}@{}}%
\column{56}{@{}>{\hspre}l<{\hspost}@{}}%
\column{62}{@{}>{\hspre}l<{\hspost}@{}}%
\column{66}{@{}>{\hspre}l<{\hspost}@{}}%
\column{75}{@{}>{\hspre}l<{\hspost}@{}}%
\column{79}{@{}>{\hspre}l<{\hspost}@{}}%
\column{E}{@{}>{\hspre}l<{\hspost}@{}}%
\>[3]{}\sigma{}\<[6]%
\>[6]{}\mathrel{\;\;::=\;\;}{}\<[6E]%
\>[11]{}\Varid{a}\;{}\<[14]%
\>[14]{}\mathrel{|}\;\sigma\mathrel{{<}\hspace{-0.4em}\star\hspace{-0.4em}{>}}\sigma\;{}\<[26]%
\>[26]{}\mathrel{|}\;{}\<[30]%
\>[30]{}\sigma\mathrel{{<}\hspace{-0.4em}\mid\hspace{-0.4em}{>}}\sigma\;{}\<[39]%
\>[39]{}\mathrel{|}\;{}\<[43]%
\>[43]{}\gamma\;{}\<[52]%
\>[52]{}\mathrel{|}\;{}\<[56]%
\>[56]{}\delta\;{}\<[62]%
\>[62]{}\mathrel{|}\;{}\<[66]%
\>[66]{}\ell\;\sigma\;{}\<[75]%
\>[75]{}\mathrel{|}\;{}\<[79]%
\>[79]{}\mu\;\Varid{f}{}\<[E]%
\\[\blanklineskip]%
\>[3]{}\Varid{a}{}\<[6]%
\>[6]{}\mathrel{\;\;::=\;\;}{}\<[6E]%
\>[11]{}\Varid{r}\;{}\<[14]%
\>[14]{}\mathrel{|}{}\<[18]%
\>[18]{}\mathord{\sim}\sigma{}\<[18E]%
\ColumnHook
\end{hscode}\resethooks
The components of the grammar for \ensuremath{\sigma} are called \emph{strategy
  combinators}~\cite{recognizingstrategies}. Two (sub)strategies can be combined
into a strategy using the sequence (\ensuremath{\mathrel{{<}\hspace{-0.4em}\star\hspace{-0.4em}{>}}}) or choice (\ensuremath{\mathrel{{<}\hspace{-0.4em}\mid\hspace{-0.4em}{>}}}) combinator, with
\ensuremath{\gamma} (always succeeds) and \ensuremath{\delta} (always fails) as unit elements,
respectively. A strategy can be tagged with a label (\ensuremath{\ell}). The \ensuremath{\mu}
combinator returns the fixed-point of a function that takes a strategy as
argument and returns a strategy. The non-terminal symbol \ensuremath{\Varid{a}} is either a rewrite
rule \ensuremath{\Varid{r}} or an applicability check \ensuremath{\mathord{\sim}\sigma} parametrised over a strategy.

\end{definition}

This definition corresponds to the definition of a context-free grammar (CFG),
extended with fixed-points, with an alphabet consisting of rewrite rules and
applicability checks. We distinguish two kinds of rules: \emph{minor} rules,
also called administrative rules, and \emph{major} (normal) rewrite rules. Minor
rules are used to perform administrative tasks, such as moving down into a term,
updating an environment, or automatically simplifying a term, such as replacing
\ensuremath{\Varid{x}\mathbin{+}\mathrm{0}} by \ensuremath{\Varid{x}}. A major rule may be turned into a minor rule to decrease the
granularity of intermediate steps (e.g., rewriting \ensuremath{\mathrm{3}\mathbin{+}\mathrm{4}} to \ensuremath{\mathrm{7}} in our running
example), or increase the difficulty of an exercise.  The function \ensuremath{\Varid{isMinor}} is
used to determine whether or not a rewrite rule is minor. When tracking a
student working on an exercise we maintain an environment, for example for
storing extra information. Major rules typically are the rules a user applies,
such as the three rules for manipulating powers given above. The example
derivation given above only shows the major rules. The minor rules which move
the focus into a term, for example for moving from \ensuremath{{({\Varid{a}^{\mathrm{3}}}\cdot{\Varid{a}^{\mathrm{4}}})^{\mathrm{2}}}} to \ensuremath{{\Varid{a}^{\mathrm{3}}}\cdot{\Varid{a}^{\mathrm{4}}}} to apply rule \ensuremath{\textsc{AddExp}}, are not shown in
the derivation. The author determines whether or not a rule is major or
minor. It is advisable to make only rules which the user can apply major, since
major rules will be shown to the user in derivations, and be given as hints. For
example, if the focus in the editor cannot be set by the user, it is unwise to
make a rule that changes the focus in a term a major rule.

The main purpose of our feedback framework is to track student behaviour, and to
automatically calculate feedback based on the strategy and the current term. For
this purpose we have to find out \emph{where} in a strategy an error is made,
for example. Error messages or hints depend on the position in the strategy at
which a student has arrived. To connect an error message to a particular position
in a strategy we \emph{label} a position in the strategy. A labeled strategy can
be transformed to a non-labeled strategy with additional minor rules, namely
\ensuremath{\textsc{Enter}} and \ensuremath{\textsc{Leave}}. The \ensuremath{\textsc{Enter}} and \ensuremath{\textsc{Leave}} minor rules update the environment
in order to keep track where we are in a strategy.

The \emph{language} of a strategy is the set of sequences of rewrite rules and 
applicability checks returned by function \ensuremath{\Conid{L}}:
\begin{hscode}\SaveRestoreHook
\column{B}{@{}>{\hspre}l<{\hspost}@{}}%
\column{3}{@{}>{\hspre}l<{\hspost}@{}}%
\column{18}{@{}>{\hspre}c<{\hspost}@{}}%
\column{18E}{@{}l@{}}%
\column{21}{@{}>{\hspre}l<{\hspost}@{}}%
\column{E}{@{}>{\hspre}l<{\hspost}@{}}%
\>[3]{}\Conid{L}\;(\Varid{a}){}\<[18]%
\>[18]{}\mathrel{=}{}\<[18E]%
\>[21]{}\{\mskip1.5mu \Varid{a}\mskip1.5mu\}{}\<[E]%
\\
\>[3]{}\Conid{L}\;(\sigma_{1}\mathrel{{<}\hspace{-0.4em}\star\hspace{-0.4em}{>}}\sigma_{2}){}\<[18]%
\>[18]{}\mathrel{=}{}\<[18E]%
\>[21]{}\{\mskip1.5mu \Varid{xy}\mid \Varid{x}\;\mathrel{\in}\;\Conid{L}\;(\sigma_{1}),\Varid{y}\;\mathrel{\in}\;\Conid{L}\;(\sigma_{2})\mskip1.5mu\}{}\<[E]%
\\
\>[3]{}\Conid{L}\;(\sigma_{1}\mathrel{{<}\hspace{-0.4em}\mid\hspace{-0.4em}{>}}\sigma_{2}){}\<[18]%
\>[18]{}\mathrel{=}{}\<[18E]%
\>[21]{}\Conid{L}\;(\sigma_{1})\cup\Conid{L}\;(\sigma_{2}){}\<[E]%
\\
\>[3]{}\Conid{L}\;(\gamma){}\<[18]%
\>[18]{}\mathrel{=}{}\<[18E]%
\>[21]{}\{\mskip1.5mu \epsilon \mskip1.5mu\}{}\<[E]%
\\
\>[3]{}\Conid{L}\;(\delta){}\<[18]%
\>[18]{}\mathrel{=}{}\<[18E]%
\>[21]{}\emptyset{}\<[E]%
\\
\>[3]{}\Conid{L}\;(\ell\;\sigma){}\<[18]%
\>[18]{}\mathrel{=}{}\<[18E]%
\>[21]{}\Conid{L}\;(\sigma){}\<[E]%
\\
\>[3]{}\Conid{L}\;(\mu\;\Varid{f}){}\<[18]%
\>[18]{}\mathrel{=}{}\<[18E]%
\>[21]{}\Conid{L}\;(\Varid{f}\;(\mu\;\Varid{f})){}\<[E]%
\ColumnHook
\end{hscode}\resethooks
For example, using Haskell notation for representing lists and omitting minor
rules, the list \ensuremath{[\mskip1.5mu \textsc{AddExp},\textsc{MulExp}\mskip1.5mu]} is an element of the language of the strategy
for solving power exercises. We will give a more detailed example in 
Section~\ref{ssec:navi}.

\subsection{Derived Combinators} \label{ssec:combinators} 
On top of the primitive strategy combinators given in the previous subsection,
we define a set of derived strategy combinators. These derived strategy
combinators are very useful for formulating rewrite strategies for exercises.
They are be built on top of the basic concepts.

An important combinator is the left-biased choice (\ensuremath{\triangleright}), which can be
defined as follows:
\begin{hscode}\SaveRestoreHook
\column{B}{@{}>{\hspre}l<{\hspost}@{}}%
\column{3}{@{}>{\hspre}l<{\hspost}@{}}%
\column{13}{@{}>{\hspre}l<{\hspost}@{}}%
\column{E}{@{}>{\hspre}l<{\hspost}@{}}%
\>[3]{}\sigma_{1}\triangleright\sigma_{2}{}\<[13]%
\>[13]{}\mathrel{=}\sigma_{1}\mathrel{{<}\hspace{-0.4em}\mid\hspace{-0.4em}{>}}(\mathord{\sim}\sigma_{1}\mathrel{{<}\hspace{-0.4em}\star\hspace{-0.4em}{>}}\sigma_{2}){}\<[E]%
\ColumnHook
\end{hscode}\resethooks
The strategy \ensuremath{\sigma_{2}} is only considered when \ensuremath{\sigma_{1}} is not applicable. Other
combinators, such as \ensuremath{\Varid{try}}, \ensuremath{\Varid{repeat}}, and \ensuremath{\Varid{option}}, are similar to the
well-known EBNF (extended Backus Naur form) constructs.
\begin{hscode}\SaveRestoreHook
\column{B}{@{}>{\hspre}l<{\hspost}@{}}%
\column{3}{@{}>{\hspre}l<{\hspost}@{}}%
\column{13}{@{}>{\hspre}l<{\hspost}@{}}%
\column{E}{@{}>{\hspre}l<{\hspost}@{}}%
\>[3]{}\Varid{option}\;\sigma{}\<[13]%
\>[13]{}\mathrel{=}\sigma\mathrel{{<}\hspace{-0.4em}\mid\hspace{-0.4em}{>}}\gamma{}\<[E]%
\\[\blanklineskip]%
\>[3]{}\Varid{try}\;\sigma{}\<[13]%
\>[13]{}\mathrel{=}\sigma\triangleright\gamma{}\<[E]%
\\[\blanklineskip]%
\>[3]{}\Varid{repeat}\;\sigma{}\<[13]%
\>[13]{}\mathrel{=}\mu\;(\lambda \Varid{x}\;.\;\Varid{try}\;(\sigma\mathrel{{<}\hspace{-0.4em}\star\hspace{-0.4em}{>}}\Varid{x})){}\<[E]%
\ColumnHook
\end{hscode}\resethooks
The \ensuremath{\Varid{option}} strategy combinator applies a strategy optionally. The \ensuremath{\Varid{try}} 
combinator applies a strategy when it is applicable. The \ensuremath{\Varid{repeat}} combinator
tries to apply a strategy as many times as possible.

\subsection{Navigation}
\label{ssec:navi}
Besides the derived combinators from the previous subsection, we add a set of
traversal combinators to our strategy language. Traversal combinators traverse a
term, and for example perform rewrite rules or strategies \ensuremath{\Varid{somewhere}} or
\ensuremath{\Varid{bottomUp}}. We use a number of administrative rules for navigating through the
abstract syntax tree (AST) of an expression: \ensuremath{\textsc{Up}}, \ensuremath{\textsc{Down}}, \ensuremath{\textsc{Left}}, and \ensuremath{\textsc{Right}}.
The minor rule \ensuremath{\textsc{Down}} takes a function as argument, which decides which child to
select based on the environment. Using \ensuremath{\textsc{Down}} we construct a minor rule that
selects all children using the choice combinator: \ensuremath{\textsc{Downs}}. Navigation is
implemented by means of a zipper~\cite{huet_zipper}, which is an efficient data
structure to define and move a focus in an expression. The zipper can be seen as
a combination of an expression and its context. An alternative way to navigate
is to use position information of (sub)expressions. An implementation using this
approach uses a list of integers denoting a path from the top of the expression
to the subexpression in focus. This approach is not as efficient and type-safe
as a zipper, since the AST needs to be traversed to retrieve the subexpression
in focus, and since it is possible to specify paths that do not correspond to a
position in the tree.

Many traversal combinators use the \ensuremath{\Varid{once}} combinator:
\begin{hscode}\SaveRestoreHook
\column{B}{@{}>{\hspre}l<{\hspost}@{}}%
\column{4}{@{}>{\hspre}l<{\hspost}@{}}%
\column{14}{@{}>{\hspre}l<{\hspost}@{}}%
\column{E}{@{}>{\hspre}l<{\hspost}@{}}%
\>[4]{}\Varid{once}\;\sigma\mathrel{=}{}\<[14]%
\>[14]{}\textsc{Downs}\mathrel{{<}\hspace{-0.4em}\star\hspace{-0.4em}{>}}\sigma\mathrel{{<}\hspace{-0.4em}\star\hspace{-0.4em}{>}}\textsc{Up}{}\<[E]%
\ColumnHook
\end{hscode}\resethooks
The \ensuremath{\Varid{once}} combinator takes a strategy as argument, and applies it to a direct
child of the expression currently in focus. After applying \ensuremath{\Varid{once}} the focus is
again at the top-level expression. The \ensuremath{\Varid{once}} combinator returns all possible
ways, by means of the choice combinator introduced by \ensuremath{\textsc{Downs}}, in which a
strategy can be applied once to a direct child of an expression. So an
application of a strategy constructed with the \ensuremath{\Varid{once}} combinator may have more
than one result, depending on whether or not strategy \ensuremath{\sigma} is applicable to the
children.

The traversal combinator \ensuremath{\Varid{somewhere}} applies a strategy to a single
subexpression (including the expression itself).
\begin{hscode}\SaveRestoreHook
\column{B}{@{}>{\hspre}l<{\hspost}@{}}%
\column{4}{@{}>{\hspre}l<{\hspost}@{}}%
\column{18}{@{}>{\hspre}c<{\hspost}@{}}%
\column{18E}{@{}l@{}}%
\column{21}{@{}>{\hspre}l<{\hspost}@{}}%
\column{E}{@{}>{\hspre}l<{\hspost}@{}}%
\>[4]{}\Varid{somewhere}\;\sigma{}\<[18]%
\>[18]{}\mathrel{=}{}\<[18E]%
\>[21]{}\mu\;(\lambda \Varid{x}\;.\;\sigma\mathrel{{<}\hspace{-0.4em}\mid\hspace{-0.4em}{>}}\Varid{once}\;\Varid{x}){}\<[E]%
\ColumnHook
\end{hscode}\resethooks
If we want to be more specific about where to apply a strategy, we can instead
use \ensuremath{\Varid{bottomUp}} or \ensuremath{\Varid{topDown}}:
\begin{hscode}\SaveRestoreHook
\column{B}{@{}>{\hspre}l<{\hspost}@{}}%
\column{3}{@{}>{\hspre}l<{\hspost}@{}}%
\column{12}{@{}>{\hspre}l<{\hspost}@{}}%
\column{15}{@{}>{\hspre}c<{\hspost}@{}}%
\column{15E}{@{}l@{}}%
\column{18}{@{}>{\hspre}l<{\hspost}@{}}%
\column{E}{@{}>{\hspre}l<{\hspost}@{}}%
\>[3]{}\Varid{bottomUp}\;\sigma{}\<[15]%
\>[15]{}\mathrel{=}{}\<[15E]%
\>[18]{}\mu\;(\lambda \Varid{x}\;.\;\Varid{once}\;\Varid{x}\triangleright\sigma){}\<[E]%
\\[\blanklineskip]%
\>[3]{}\Varid{topDown}\;{}\<[12]%
\>[12]{}\sigma{}\<[15]%
\>[15]{}\mathrel{=}{}\<[15E]%
\>[18]{}\mu\;(\lambda \Varid{x}\;.\;\sigma\triangleright\Varid{once}\;\Varid{x}){}\<[E]%
\ColumnHook
\end{hscode}\resethooks
These combinators search for a suitable location to apply an argument strategy
in a bottom-up or top-down fashion, without imposing an order in which the
children are visited. These combinators do not apply their argument strategy 
exhaustively, instead, the argument strategy is applied only once.

Navigation operators navigate through the abstract syntax of the domain on which
the rewrite rules are specified. The example rules we presented at the beginning 
of this section work on algebraic expressions containing powers. We give a
formal definition of the power domain, and the zipper we use for navigation on 
this domain.
\begin{hscode}\SaveRestoreHook
\column{B}{@{}>{\hspre}l<{\hspost}@{}}%
\column{3}{@{}>{\hspre}l<{\hspost}@{}}%
\column{8}{@{}>{\hspre}c<{\hspost}@{}}%
\column{8E}{@{}l@{}}%
\column{13}{@{}>{\hspre}l<{\hspost}@{}}%
\column{E}{@{}>{\hspre}l<{\hspost}@{}}%
\>[3]{}\Varid{e}{}\<[8]%
\>[8]{}\mathrel{\;\;::=\;\;}{}\<[8E]%
\>[13]{}\Varid{v}\;\mathrel{|}\;{\Varid{e}^{\Varid{n}}}\;\mathrel{|}\;\Varid{e}\cdot\Varid{e}{}\<[E]%
\\
\>[3]{}\varphi {}\<[8]%
\>[8]{}\mathrel{\;\;::=\;\;}{}\<[8E]%
\>[13]{}{\llbracket\Varid{e}\rrbracket}\;\mathrel{|}\;{\varphi ^{\Varid{n}}}\;\mathrel{|}\;\varphi \cdot\Varid{e}\;\mathrel{|}\;\Varid{e}\cdot\varphi {}\<[E]%
\ColumnHook
\end{hscode}\resethooks
Here \ensuremath{\Varid{v}} stands for a variable and \ensuremath{\Varid{n}} for an integer number. In the grammar for
the zipper, the expression between double square brackets is the expression in
focus. The focus can appear inside a power \ensuremath{{\varphi ^{\Varid{n}}}}, or move left or right
into a product (\ensuremath{\varphi \cdot\Varid{e}}, or \ensuremath{\Varid{e}\cdot\varphi }). The minor rules for navigation are
defined by analysing the various forms of the zipper. For example, the \ensuremath{\textsc{Up}} rule
is defined as follows:
\begin{hscode}\SaveRestoreHook
\column{B}{@{}>{\hspre}l<{\hspost}@{}}%
\column{3}{@{}>{\hspre}l<{\hspost}@{}}%
\column{18}{@{}>{\hspre}l<{\hspost}@{}}%
\column{26}{@{}>{\hspre}l<{\hspost}@{}}%
\column{47}{@{}>{\hspre}l<{\hspost}@{}}%
\column{E}{@{}>{\hspre}l<{\hspost}@{}}%
\>[3]{}\textsc{Up}\;\varphi \mathrel{=}\mathbf{case}\;{}\<[18]%
\>[18]{}\varphi \;\mathbf{of}\;{}\<[26]%
\>[26]{}{{\llbracket\Varid{e}\rrbracket}^{\Varid{i}}}{}\<[47]%
\>[47]{}\rightarrow{\llbracket{\Varid{e}^{\Varid{i}}}\rrbracket}{}\<[E]%
\\
\>[26]{}{\llbracket\Varid{e}_{1}\rrbracket}\cdot\Varid{e}_{2}{}\<[47]%
\>[47]{}\rightarrow{\llbracket\Varid{e}_{1}\cdot\Varid{e}_{2}\rrbracket}{}\<[E]%
\\
\>[26]{}\Varid{e}_{1}\cdot{\llbracket\Varid{e}_{2}\rrbracket}{}\<[47]%
\>[47]{}\rightarrow{\llbracket\Varid{e}_{1}\cdot\Varid{e}_{2}\rrbracket}{}\<[E]%
\\
\>[26]{}{\varphi ^{\Varid{n}}}{}\<[47]%
\>[47]{}\rightarrow{(\textsc{Up}\;\varphi )^{\Varid{n}}}{}\<[E]%
\\
\>[26]{}\varphi \cdot\Varid{e}{}\<[47]%
\>[47]{}\rightarrow\textsc{Up}\;\varphi \cdot\Varid{e}{}\<[E]%
\\
\>[26]{}\Varid{e}\cdot\varphi {}\<[47]%
\>[47]{}\rightarrow\Varid{e}\cdot\textsc{Up}\;\varphi {}\<[E]%
\ColumnHook
\end{hscode}\resethooks

We return to the example we gave at the beginning of this section. Now that we
have a formal definition of strategies, we can express the informal strategy in
terms of the derived strategy combinators. 
\begin{hscode}\SaveRestoreHook
\column{B}{@{}>{\hspre}l<{\hspost}@{}}%
\column{3}{@{}>{\hspre}l<{\hspost}@{}}%
\column{E}{@{}>{\hspre}l<{\hspost}@{}}%
\>[3]{}\Varid{writeAsPowerOf}\mathrel{=}\ell\;(\Varid{repeat}\;(\Varid{bottomUp}\;(\textsc{AddExp}\mathrel{{<}\hspace{-0.4em}\mid\hspace{-0.4em}{>}}\textsc{MulExp}\mathrel{{<}\hspace{-0.4em}\mid\hspace{-0.4em}{>}}\textsc{DistExp}))){}\<[E]%
\ColumnHook
\end{hscode}\resethooks
Applying the \ensuremath{\Varid{writeAsPowerOf}} strategy to the expression \ensuremath{{({\Varid{a}^{\mathrm{3}}}\cdot{\Varid{a}^{\mathrm{4}}})^{\mathrm{2}}}} gives the following derivation:
\begin{hscode}\SaveRestoreHook
\column{B}{@{}>{\hspre}l<{\hspost}@{}}%
\column{3}{@{}>{\hspre}l<{\hspost}@{}}%
\column{E}{@{}>{\hspre}l<{\hspost}@{}}%
\>[3]{}{\llbracket{({\Varid{a}^{\mathrm{3}}}\cdot{\Varid{a}^{\mathrm{4}}})^{\mathrm{2}}}\rrbracket}\;{\; \stackrel{(\textsc{Enter}\;\ell)}{\Longrightarrow} \;}\;{\llbracket{({\Varid{a}^{\mathrm{3}}}\cdot{\Varid{a}^{\mathrm{4}}})^{\mathrm{2}}}\rrbracket}\;{\; \stackrel{\textsc{Down}}{\Longrightarrow} \;}\;{{\llbracket{\Varid{a}^{\mathrm{3}}}\cdot{\Varid{a}^{\mathrm{4}}}\rrbracket}^{\mathrm{2}}}\;{\; \stackrel{\textsc{AppCheck}}{\Longrightarrow} \;}\;{{\llbracket{\Varid{a}^{\mathrm{3}}}\cdot{\Varid{a}^{\mathrm{4}}}\rrbracket}^{\mathrm{2}}}\;{\; \stackrel{\textsc{AddExp}}{\Longrightarrow} \;}{}\<[E]%
\\[\blanklineskip]%
\>[3]{}{{\llbracket{\Varid{a}^{\mathrm{7}}}\rrbracket}^{\mathrm{2}}}\;{\; \stackrel{\textsc{Up}}{\Longrightarrow} \;}\;{\llbracket{({\Varid{a}^{\mathrm{7}}})^{\mathrm{2}}}\rrbracket}\;{\; \stackrel{\textsc{AppCheck}}{\Longrightarrow} \;}\;{\llbracket{({\Varid{a}^{\mathrm{7}}})^{\mathrm{2}}}\rrbracket}\;{\; \stackrel{\textsc{MulExp}}{\Longrightarrow} \;}\;{\llbracket{\Varid{a}^{\mathrm{14}}}\rrbracket}\;{\; \stackrel{\textsc{AppCheck}}{\Longrightarrow} \;}\;{\llbracket{\Varid{a}^{\mathrm{14}}}\rrbracket}\;{\; \stackrel{(\textsc{Leave}\;\ell)}{\Longrightarrow} \;}\;{\llbracket{\Varid{a}^{\mathrm{14}}}\rrbracket}{}\<[E]%
\ColumnHook
\end{hscode}\resethooks
The \ensuremath{\textsc{AppCheck}} (applicability check) is introduced by the \ensuremath{\Varid{repeat}} and
\ensuremath{\Varid{bottomUp}} combinator. The navigation rules do not work directly on an
expression, but on the zipper containing the expression (\ensuremath{\varphi }). As a
consequence, if a strategy uses a traversal combinator it is only applicable to
an expression in a context. (In the example derivation above we omit the context
from the rewrite steps, and only show the expression and the focus.) To keep the
definition of rewrite rules simple and clean, we have a function that lifts a
rewrite rule to operate on an expression in a context. The definition of this
function is omitted.

\subsection{Semantics}
\label{ssec:semantics}
At the start of this section we defined the language of a strategy. The language
contains lists of rewrite rules and applicability conditions. In this subsection
we define how a strategy transforms an expression. We use the following
terminology in the definition of the semantics:
\begin{itize}
\item[{\it Expression.}] We only consider expressions generated by a grammar.
\item[{\it Environment.}] An environment stores additional information at
  intermediate steps, such as auxiliary results. An environment is a set of
  key/value pairs, which can be added, removed, consulted, and updated.
\item[{\it Rewrite rule.}] A rewrite rule \ensuremath{\Varid{r}} is a binary relation on the
  product of an environment and an expression: \ensuremath{(\Gamma_{1}\mathrel{\times}\Varid{e}_{1})\;{\buildrel \Varid{r} \over \leadsto}\;(\Gamma_{2}\mathrel{\times}\Varid{e}_{2})}. A
  rewrite rule is tagged with a boolean indicating whether or not it is a minor
  rule. If the rewrite rule does not change the environment we use the notation
  \ensuremath{\Varid{e}_{1}\;{\buildrel \Varid{r} \over \leadsto}\;\Varid{e}_{2}}.

  A lifted rewrite rule is a binary relation on the product of an environment
  and a zipper (an expression in focus together with its context): \ensuremath{(\Gamma_{1}\mathrel{\times}\varphi _{1})\;{\buildrel \Varid{r} \over \leadsto}\;(\Gamma_{2}\mathrel{\times}\varphi _{2})}.
\item[{\it State.}] The state of an exercise is modelled as the product of an
  environment, a zipper, and a strategy.
\end{itize}

To check if we have reached an end state we want to determine whether or not the
empty sentence is a member of the language of a strategy. For this purpose we 
define a function \ensuremath{\Varid{empty}}, similar to the function \ensuremath{\Varid{empty}} defined for CFGs. We 
first define an auxiliary function that returns all sentences that consist of 
minor rules only:
\begin{definition}
  Function \ensuremath{\Varid{minorSentences}} returns all sentences in the language of a
  strategy that consist of minor rules only, including the empty sentence
  \ensuremath{\epsilon }:
  \begin{hscode}\SaveRestoreHook
\column{B}{@{}>{\hspre}l<{\hspost}@{}}%
\column{3}{@{}>{\hspre}l<{\hspost}@{}}%
\column{40}{@{}>{\hspre}c<{\hspost}@{}}%
\column{40E}{@{}l@{}}%
\column{43}{@{}>{\hspre}l<{\hspost}@{}}%
\column{67}{@{}>{\hspre}l<{\hspost}@{}}%
\column{E}{@{}>{\hspre}l<{\hspost}@{}}%
\>[3]{}\Varid{minorSentences}\;(\sigma)\mathrel{=}\{\mskip1.5mu \Varid{r}_{1}\;\ldots\;\mathit{r_n}{}\<[40]%
\>[40]{}\mid {}\<[40E]%
\>[43]{}\Varid{r}_{1}\;\ldots\;\mathit{r_n}\;\mathrel{\in}\;\Conid{L}\;(\sigma),{}\<[67]%
\>[67]{}\forall\;\Varid{i}\;\mathrel{\in}\;\mathrm{1}\;\ldots\;\Varid{n}\;.\;\Varid{isMinor}\;(\mathit{r_i})\mskip1.5mu\}{}\<[E]%
\ColumnHook
\end{hscode}\resethooks
\end{definition}
\noindent
We need this function to determine if there are any trailing minor rules after
the last major rewrite rule. We define the \ensuremath{\Varid{empty}} function in terms of the
\ensuremath{\Varid{minorSentences}} function:
\begin{definition}
  Function \ensuremath{\Varid{empty}} checks whether or not the language of a strategy contains 
  the empty sentence \ensuremath{\epsilon }, or a sentence consisting of minor rules only:
  \begin{hscode}\SaveRestoreHook
\column{B}{@{}>{\hspre}l<{\hspost}@{}}%
\column{3}{@{}>{\hspre}l<{\hspost}@{}}%
\column{E}{@{}>{\hspre}l<{\hspost}@{}}%
\>[3]{}\Varid{empty}\;(\sigma)\mathrel{=}\Varid{minorSentences}\;(\sigma)\;\neq\;\emptyset{}\<[E]%
\ColumnHook
\end{hscode}\resethooks
\end{definition}
\noindent
Both functions are easily lifted to take a state as argument.

The smallest action that can be performed with a strategy is a \emph{step}: the
application of a rewrite rule. Before we define a step relation that performs a
state transition, we define a relation that splits a strategy into its first
rule or applicability check and the remaining strategy.

\begin{definition}
\label{def:split}
  The relation \ensuremath{\mapsto} splits a strategy into a rule or an applicability check and 
  the remaining strategy: \ensuremath{\sigma_{1}\mapsto\Varid{a}\mathrel{{<}\hspace{-0.4em}\star\hspace{-0.4em}{>}}\sigma_{2}}.  
  \renewcommand{\arraystretch}{0.2}

  \begin{tabular}{P}
    \ensuremath{\Varid{a}\mapsto\Varid{a}\mathrel{{<}\hspace{-0.4em}\star\hspace{-0.4em}{>}}\gamma} \quad  
    \infrule{\ensuremath{\sigma_{1}\mapsto\Varid{a}\mathrel{{<}\hspace{-0.4em}\star\hspace{-0.4em}{>}}\sigma_{3}}}{\ensuremath{\sigma_{1}\mathrel{{<}\hspace{-0.4em}\star\hspace{-0.4em}{>}}\sigma_{2}\mapsto\Varid{a}\mathrel{{<}\hspace{-0.4em}\star\hspace{-0.4em}{>}}(\sigma_{3}\mathrel{{<}\hspace{-0.4em}\star\hspace{-0.4em}{>}}\sigma_{2})}} \quad
    \infrule{\ensuremath{\epsilon \;\mathrel{\in}\;\Conid{L}\;(\sigma_{1})\;\quad\;\sigma_{2}\mapsto\Varid{a}\mathrel{{<}\hspace{-0.4em}\star\hspace{-0.4em}{>}}\sigma_{3}}}{\ensuremath{\sigma_{1}\mathrel{{<}\hspace{-0.4em}\star\hspace{-0.4em}{>}}\sigma_{2}\mapsto\Varid{a}\mathrel{{<}\hspace{-0.4em}\star\hspace{-0.4em}{>}}\sigma_{3}}} \
    \\
    \infrule{\ensuremath{\sigma_{1}\mapsto\Varid{a}\mathrel{{<}\hspace{-0.4em}\star\hspace{-0.4em}{>}}\sigma_{3}}}{\ensuremath{\sigma_{1}\mathrel{{<}\hspace{-0.4em}\mid\hspace{-0.4em}{>}}\sigma_{2}\rightarrow\Varid{a}\mathrel{{<}\hspace{-0.4em}\star\hspace{-0.4em}{>}}\sigma_{3}}} \quad
    \infrule{\ensuremath{\sigma_{2}\mapsto\Varid{a}\mathrel{{<}\hspace{-0.4em}\star\hspace{-0.4em}{>}}\sigma_{3}}}{\ensuremath{\sigma_{1}\mathrel{{<}\hspace{-0.4em}\mid\hspace{-0.4em}{>}}\sigma_{2}\rightarrow\Varid{a}\mathrel{{<}\hspace{-0.4em}\star\hspace{-0.4em}{>}}\sigma_{3}}} \quad
    \infrule{\ensuremath{\Varid{f}\;(\mu\;\Varid{f})\mapsto\Varid{a}\mathrel{{<}\hspace{-0.4em}\star\hspace{-0.4em}{>}}\sigma}}{\ensuremath{\mu\;\Varid{f}\mapsto\Varid{a}\mathrel{{<}\hspace{-0.4em}\star\hspace{-0.4em}{>}}\sigma}}   
    \\
    \ensuremath{\ell\;\sigma\mapsto\textsc{Enter}\;\ell\mathrel{{<}\hspace{-0.4em}\star\hspace{-0.4em}{>}}(\sigma\mathrel{{<}\hspace{-0.4em}\star\hspace{-0.4em}{>}}\textsc{Leave}\;\ell)}
  \end{tabular}
  
\end{definition}

\begin{definition}
  The step operator \ensuremath{{\buildrel \Varid{r} \over \rightarrow}} denotes the relation between the current state \ensuremath{\Conid{S}_{1}}
  and a new state \ensuremath{\Conid{S}_{2}} (obtained by applying the rewrite rule \ensuremath{\Varid{r}}).

  \noindent
  \begin{tabular}{P}
    \infrule{\ensuremath{\sigma_{1}\mapsto\Varid{r}\mathrel{{<}\hspace{-0.4em}\star\hspace{-0.4em}{>}}\sigma_{2}\;\quad\;(\Gamma_{1}\mathrel{\times}\varphi _{1})\;{\buildrel \Varid{r} \over \leadsto}\;(\Gamma_{2}\mathrel{\times}\varphi _{2})}}
            {\ensuremath{(\Gamma_{1}\mathrel{\times}\varphi _{1}\mathrel{\times}\sigma_{1})\;{\buildrel \Varid{r} \over \rightarrow}\;(\Gamma_{2}\mathrel{\times}\varphi _{2}\mathrel{\times}\sigma_{2})}}
    \quad
    \infrule{\ensuremath{\sigma_{1}\mapsto\mathord{\sim}\sigma_{2}\mathrel{{<}\hspace{-0.4em}\star\hspace{-0.4em}{>}}\sigma_{3}\;\quad\;\Varid{run}\;(\Gamma_{1}\mathrel{\times}\varphi _{1}\mathrel{\times}\sigma_{2})\mathrel{=}\emptyset}\\[\infruledist]
             \ensuremath{(\Gamma_{1}\mathrel{\times}\varphi _{1}\mathrel{\times}\sigma_{3})\;{\buildrel \Varid{r} \over \rightarrow}\;(\Gamma_{2}\mathrel{\times}\varphi _{2}\mathrel{\times}\sigma_{4})}}
           {\ensuremath{(\Gamma_{1}\mathrel{\times}\varphi _{1}\mathrel{\times}\sigma_{1})\;{\buildrel \Varid{r} \over \rightarrow}\;(\Gamma_{2}\mathrel{\times}\varphi _{2}\mathrel{\times}\sigma_{4})}}
  \end{tabular}
  
\end{definition}
The \ensuremath{\Varid{run}} function used in the last rule applies a strategy to a term in a 
context; its definition is given below.

The step relation \ensuremath{{\buildrel \Varid{r} \over \rightarrow}} ignores whether or not a rule is minor, and it deals
with minor and major rewrite rules in the same way. Many of the feedback 
services we have defined ignore minor rewrite rules, since we do not want to 
show such administrative steps to a user. We define a relation similar to step, 
which ignores minor rewrite rules.

\begin{definition}
  The big step operator \ensuremath{{\buildrel \Varid{r} \over \twoheadrightarrow}} denotes the relation between a state \ensuremath{\Conid{S}_{1}} and
  a new state \ensuremath{\Conid{S}_{2}}. The new state is obtained by possibly applying minor rules,
  followed by the application of a single major rewrite rule \ensuremath{\Varid{r}}. If \ensuremath{\Varid{r}} is the
  last major rule then trailing minor rules, if any, are applied as well.

  \renewcommand{\arraystretch}{0.2}  

  \begin{tabular}{P}
    \infrule{\ensuremath{\Conid{S}_{1}\;\stackrel{\Varid{r}_{1}}{\rightarrow}\;\Conid{S}_{2}\;\quad\;\Varid{isMinor}\;(\Varid{r}_{1})\;\quad\;\Conid{S}_{2}\;\stackrel{\Varid{r}_{2}}{\twoheadrightarrow}\;\Conid{S}_{3}}}
            {\ensuremath{\Conid{S}_{1}\;\stackrel{\Varid{r}_{2}}{\twoheadrightarrow}\;\Conid{S}_{3}}}
    \\
    \infrule{\ensuremath{\Conid{S}_{1}\;{\buildrel \Varid{r} \over \rightarrow}\;\Conid{S}_{2}\;\quad\;\neg \;\Varid{isMinor}\;(\Varid{r})\;\quad\;\Varid{minorSentences}\;(\Conid{S}_{2})\mathrel{=}\emptyset}}
            {\ensuremath{\Conid{S}_{1}\;{\buildrel \Varid{r} \over \twoheadrightarrow}\;\Conid{S}_{2}}}
    \\
   \infrule{\ensuremath{\Conid{S}_{1}\;{\buildrel \Varid{r} \over \rightarrow}\;\Conid{S}_{2}\;\quad\;\neg \;\Varid{isMinor}\;(\Varid{r})\;\quad\;\vec{\Varid{m}}\;\mathrel{\in}\;\Varid{minorSentences}\;(\Conid{S}_{2})\;\quad\;\Conid{S}_{2}\;\stackrel{\vec{\Varid{m}}}{\rightarrow}\;\Conid{S}_{3}}}
            {\ensuremath{\Conid{S}_{1}\;{\buildrel \Varid{r} \over \twoheadrightarrow}\;\Conid{S}_{3}}}
  \end{tabular}
  where \ensuremath{\vec{\Varid{m}}} is a sequence of minor rules and \ensuremath{\stackrel{\vec{\Varid{m}}}{\rightarrow}} is the
  sequential application thereof.
 
\end{definition}

It is important, when performing a big step, that trailing minor rules are
applied. The application of trailing minor rules ensure that exhaustively
applying the step or big step operator on a term, will end up in the same end
state(s). This enables us to keep the generated feedback consistent.

\begin{definition}
  The \ensuremath{\Varid{run}} function is the closure of \ensuremath{{\buildrel \Varid{r} \over \twoheadrightarrow}}, and generates all possible end
  states from a begin state \ensuremath{\Conid{S}_{0}}. An end state contains the empty strategy:
  \ensuremath{\gamma}.
  \begin{hscode}\SaveRestoreHook
\column{B}{@{}>{\hspre}l<{\hspost}@{}}%
\column{3}{@{}>{\hspre}l<{\hspost}@{}}%
\column{E}{@{}>{\hspre}l<{\hspost}@{}}%
\>[3]{}\Varid{run}\;(\Conid{S}_{0})\mathrel{=}\{\mskip1.5mu (\Gamma\mathrel{\times}\varphi \mathrel{\times}\gamma)\mid \Conid{S}_{0}\;\stackrel{\star}{\twoheadrightarrow}\;(\Gamma\mathrel{\times}\varphi \mathrel{\times}\gamma)\mskip1.5mu\}{}\<[E]%
\ColumnHook
\end{hscode}\resethooks
  where \ensuremath{\stackrel{\star}{\twoheadrightarrow}} is the transitive closure of \ensuremath{{\twoheadrightarrow}}.
  
\end{definition}

\subsection{Non-determinism}
\label{ssec:non-determinism}
Almost all exercises can be solved in several, correct ways. For example,
consider the expression \ensuremath{{({\Varid{a}^{\mathrm{3}}}\cdot{\Varid{a}^{\mathrm{4}}})^{\mathrm{2}}}} again, and
suppose it should be solved with the strategy \ensuremath{\Varid{writeAsPowerOf'}} which is
obtained by replacing \ensuremath{\Varid{bottomUp}} by \ensuremath{\Varid{somewhere}} in the strategy
\ensuremath{\Varid{writeAsPowerOf}}.
\begin{hscode}\SaveRestoreHook
\column{B}{@{}>{\hspre}l<{\hspost}@{}}%
\column{3}{@{}>{\hspre}l<{\hspost}@{}}%
\column{E}{@{}>{\hspre}l<{\hspost}@{}}%
\>[3]{}\Varid{writeAsPowerOf'}\mathrel{=}\ell\;(\Varid{repeat}\;(\Varid{somewhere}\;(\textsc{AddExp}\mathrel{{<}\hspace{-0.4em}\mid\hspace{-0.4em}{>}}\textsc{MulExp}\mathrel{{<}\hspace{-0.4em}\mid\hspace{-0.4em}{>}}\textsc{DistExp}))){}\<[E]%
\ColumnHook
\end{hscode}\resethooks
One of the questions we want to be able to answer is: what is the next step in
order to solve the exercise? This type of feedback is given by the \ensuremath{\Varid{onefirst}}
feedback service, which we will describe in more detail in the next section. In
the above example there are two possibilities: \ensuremath{\textsc{DistExp}} can be applied to the
entire expression, and \ensuremath{\textsc{AddExp}} is applicable to the subexpression \ensuremath{{\Varid{a}^{\mathrm{3}}}\cdot{\Varid{a}^{\mathrm{4}}}}. Both steps are correct, but which do we choose? Making a random
choice would make our feedback framework non-deterministic. This problem may
show up whenever we use the choice combinator, which may also be introduced by
other combinators (such as \ensuremath{\textsc{Downs}}), to combine various solution
strategies. Applying a strategy that uses the choice combinator to an expression
may result in multiple answers.

To prevent non-deterministic behaviour we introduce a rule ordering: a total
order, denoted by \ensuremath{\mathbin{<}}, on rules. For now, the ordering is only on the rules
themselves, but one can imagine taking the environment into account in the rule
ordering. An example rule ordering for the power domain is: \ensuremath{\textsc{AddExp}\mathbin{<}\textsc{MulExp}\mathbin{<}\textsc{DistExp}}.
In this case, the first step in our example is \ensuremath{\textsc{AddExp}}.

\subsection{Restrictions}
\label{ssec:restrictions}
To use strategies to track student behaviour and give feedback, we impose some
restrictions on the form of strategies. These restrictions are similar to some
of the restrictions imposed on context-free grammars to be able to use them for
parsing.

\paragraph{Left-recursion.}
A context-free grammar is left-recursive if it contains a nonterminal that can 
be rewritten in one or more steps using the productions of the grammar to a 
sequence of symbols that starts with the same nonterminal. The same definition
applies to strategies. For example, the following strategy is left-recursive:
\begin{hscode}\SaveRestoreHook
\column{B}{@{}>{\hspre}l<{\hspost}@{}}%
\column{3}{@{}>{\hspre}l<{\hspost}@{}}%
\column{E}{@{}>{\hspre}l<{\hspost}@{}}%
\>[3]{}\Varid{leftRecur}\mathrel{=}\mu\;(\lambda \Varid{x}\;.\;\Varid{x}\mathrel{{<}\hspace{-0.4em}\star\hspace{-0.4em}{>}}\textsc{AddExp}){}\<[E]%
\ColumnHook
\end{hscode}\resethooks
The left-recursion is obvious in this strategy, since \ensuremath{\Varid{x}} is in the leftmost
position in the body of the abstraction. Left-recursion is not always this easy
to spot. Strategies with leading minor rules may or may not be
left-recursive. Strictly speaking, these strategies are not left-recursive
because the strategy grammar does not differentiate between minor and major
rules. However, our semantics make it that these strategies sometimes display
left-recursive behaviour. For example:
\begin{hscode}\SaveRestoreHook
\column{B}{@{}>{\hspre}l<{\hspost}@{}}%
\column{3}{@{}>{\hspre}l<{\hspost}@{}}%
\column{E}{@{}>{\hspre}l<{\hspost}@{}}%
\>[3]{}\Varid{leftRecur'}\mathrel{=}\mu\;(\lambda \Varid{x}\;.\;\textsc{Down}\mathrel{{<}\hspace{-0.4em}\star\hspace{-0.4em}{>}}\Varid{x}\mathrel{{<}\hspace{-0.4em}\star\hspace{-0.4em}{>}}\textsc{AddExp}){}\<[E]%
\ColumnHook
\end{hscode}\resethooks
Here, applying the minor rule \ensuremath{\textsc{Down}} repeatedly will eventually cause the 
strategy to reach the leaf of an expression tree, and stop. Hence this strategy 
is not left-recursive. However, this is a property of \ensuremath{\textsc{Down}} that is not shared
by all other minor rules. If we use a minor rule that increases a counter in the
environment, an action that will always succeed, the strategy is left-recursive.

Top-down recursive parsing using a left-recursive context-free grammar is
difficult. We have chosen for top-down recursive parsing because we need to
parse a derivation incrementally. Other parsing schemes, such as bottom-up
parsing, also have their difficulties. A grammar represented in parser
combinators~\cite{hutton92} is not allowed to be left-recursive. Similarly, for
a strategy to be used in our framework, it should not be left-recursive. Trying
to determine the next possible symbol(s) of a left-recursive strategy by means
of the split operator will get stuck in a loop.

Left-recursive strategies are not the only source of non-terminating strategy
calculations. The fact that our strategy language has a fixed-point combinator
implies that we are vulnerable to non-termination. The implementation of our
strategy language has been augmented with `time-outs' that will stop the
execution and report an error message.

\paragraph{Left-factorisation.}
Left-factoring is a grammar transformation that is useful when two productions
for the same nonterminal start with the same sequence of terminal and/or
nonterminal symbols. This transformation factors out the common part of such
productions. In a strategy, the equivalent transformation factors out common
sequences of rewrite rules from sub-strategies separated by the choice
combinator.

A strategy that contains left-factors may be problematic. Consider the
following, somewhat contrived, strategy:
\begin{hscode}\SaveRestoreHook
\column{B}{@{}>{\hspre}l<{\hspost}@{}}%
\column{3}{@{}>{\hspre}l<{\hspost}@{}}%
\column{E}{@{}>{\hspre}l<{\hspost}@{}}%
\>[3]{}\Varid{leftStrat}\mathrel{=}\ell_1\;(\textsc{AddExp}\mathrel{{<}\hspace{-0.4em}\star\hspace{-0.4em}{>}}\textsc{MulExp})\mathrel{{<}\hspace{-0.4em}\mid\hspace{-0.4em}{>}}\ell_2\;(\textsc{AddExp}\mathrel{{<}\hspace{-0.4em}\star\hspace{-0.4em}{>}}\textsc{DistExp}){}\<[E]%
\ColumnHook
\end{hscode}\resethooks
\noindent The two sub-strategies labeled with \ensuremath{\ell_1} and \ensuremath{\ell_2} have a left-factor:
the rewrite rule \ensuremath{\textsc{AddExp}}. After the application of \ensuremath{\textsc{AddExp}}, we have to decide
which sub-strategy to follow. Either we follow sub-strategy \ensuremath{\ell_1}, or we follow
sub-strategy \ensuremath{\ell_2}. Committing to a choice after seeing an application \ensuremath{\textsc{AddExp}} is
unfortunate, since it will force the student to follow the same sub-strategy. For
example, if \ensuremath{\ell_1} is chosen after an application of \ensuremath{\textsc{AddExp}} has been submitted,
and a student subsequently takes the \ensuremath{\textsc{DistExp}} step, we erroneously report that
the step does not follow the strategy. Therefore, left-factorising a strategy is
desirable, since we do not want to commit early to a particular sub-strategy. The
example strategy can be left-factored as follows:
\begin{hscode}\SaveRestoreHook
\column{B}{@{}>{\hspre}l<{\hspost}@{}}%
\column{4}{@{}>{\hspre}l<{\hspost}@{}}%
\column{E}{@{}>{\hspre}l<{\hspost}@{}}%
\>[4]{}\Varid{leftStrat'}\mathrel{=}\textsc{AddExp}\mathrel{{<}\hspace{-0.4em}\star\hspace{-0.4em}{>}}(\textsc{MulExp}\mathrel{{<}\hspace{-0.4em}\mid\hspace{-0.4em}{>}}\textsc{DistExp}){}\<[E]%
\ColumnHook
\end{hscode}\resethooks
It is clear how to left-factor (major) rewrite rules, but how should we deal
with labels, or minor rules in general? In the remainder of this section we
focus on how to deal with labels. Pushing the labels inside the choice
combinator,
\begin{hscode}\SaveRestoreHook
\column{B}{@{}>{\hspre}l<{\hspost}@{}}%
\column{4}{@{}>{\hspre}l<{\hspost}@{}}%
\column{E}{@{}>{\hspre}l<{\hspost}@{}}%
\>[4]{}\Varid{leftStrat'}\mathrel{=}\textsc{AddExp}\mathrel{{<}\hspace{-0.4em}\star\hspace{-0.4em}{>}}(\ell_1\;\textsc{MulExp}\mathrel{{<}\hspace{-0.4em}\mid\hspace{-0.4em}{>}}\ell_2\;\textsc{DistExp}){}\<[E]%
\ColumnHook
\end{hscode}\resethooks
or making a choice between the two labels probably breaks the relation between
the label and the strategy. Recall that labels are used to mark positions in a
strategy, which is connected to a feedback text.

A different way to deal with left-factors uses backtracking. With backtracking
we remember the start position of the left-factor in the strategy. In the case
that the chosen sub-strategy fails we roll back and continue with a different
sub-strategy. However, using backtracking it is possible to guide a student to a
dead end (a sub-strategy that fails). This violates our goals of giving proper
and relevant feedback.

Our solution is to detect and report left-factors. The detection of left-factors
in a strategy is relatively straightforward. The detection and notification
enables strategy developers to construct strategies without left-factors.

\section{Services}
\label{sec:services}
Our feedback services offer a wide variety of feedback functionality, so that
learning environments using our services can provide different kinds of
feedback. A learning environment can amongst others ask for the following kinds
of feedback: Is the student still on the right path towards a solution? Does the
step made by the student follow the strategy for the exercise? What is the next
step the student should take? Has the student made a common error? What does a
worked-out solution look like? This section formalises the corresponding 
services.

\subsection{Exercises}
\label{ssec:exercise}
Most of the feedback services calculate feedback based on a strategy, but
sometimes there are other components that play a role in our services. All
components necessary for our feedback services are grouped together in an
\emph{exercise}. An exercise contains all domain-specific (and
exercise-specific) functionality, together with an exercise code for
identification. The most important component of an exercise is its
strategy. Additional rewrite rules can be added to an exercise to help detect
possible detours. We not only specify proper rewrite rules, but also \emph{buggy
  rules}. A buggy rule captures a common misconception. If we detect an
application of a buggy rule, we report this to the user. We also need predicates
for checking whether or not an expression is a suitable starting expression that
can be solved by the strategy, and whether or not an expression is in solved
form. For diagnosing intermediate answers, we need an equivalence relation to
compare a submission with a preceding expression, and a similarity relation,
which is possibly more liberal than syntactic equality. Although not of primary
importance, it can be convenient to have a randomised expression generator for
the exercise. The last component of an exercise is a function that returns the
order of rules.

\begin{definition}
  An exercise consists of an identification code, a strategy \ensuremath{\sigma}, a rule set \ensuremath{\{\mskip1.5mu \Varid{r}_{1}\;\ldots\;\mathit{r_n}\mskip1.5mu\}}, a buggy rule set, an equivalence relation \ensuremath{\equiv }, a similarity
  relation \ensuremath{\approx}, a predicate \ensuremath{\Varid{isSuitable}}, a predicate \ensuremath{\Varid{isReady}}, an
  optional expression generator, and a rule ordering function \ensuremath{\mathbin{<}}.
\end{definition}
\noindent
An example of an exercise is the exercise for our running example: calculating 
with powers. 

\begin{hscode}\SaveRestoreHook
\column{B}{@{}>{\hspre}l<{\hspost}@{}}%
\column{3}{@{}>{\hspre}l<{\hspost}@{}}%
\column{21}{@{}>{\hspre}c<{\hspost}@{}}%
\column{21E}{@{}l@{}}%
\column{24}{@{}>{\hspre}l<{\hspost}@{}}%
\column{34}{@{}>{\hspre}l<{\hspost}@{}}%
\column{60}{@{}>{\hspre}l<{\hspost}@{}}%
\column{76}{@{}>{\hspre}l<{\hspost}@{}}%
\column{E}{@{}>{\hspre}l<{\hspost}@{}}%
\>[3]{}\Varid{powerExercise}\mathrel{=}{}\<[21]%
\>[21]{}({}\<[21E]%
\>[24]{}{\tt powerExercise}{}\<[E]%
\\
\>[21]{},{}\<[21E]%
\>[24]{}\ell\;(\Varid{repeat}\;(\Varid{bottomUp}\;(\textsc{AddExp}\mathrel{{<}\hspace{-0.4em}\mid\hspace{-0.4em}{>}}\textsc{MulExp}\mathrel{{<}\hspace{-0.4em}\mid\hspace{-0.4em}{>}}\textsc{DistExp}))){}\<[E]%
\\
\>[21]{},{}\<[21E]%
\>[24]{}\{\mskip1.5mu {\Varid{a}^{\Varid{x}}}\mathrel{=}\frac{\mathrm{1}}{{\Varid{a}^{\mathbin{-}\Varid{x}}}}\;\quad\;[\mskip1.5mu \textsc{ReciExp}\mskip1.5mu]\mskip1.5mu\}{}\<[E]%
\\
\>[21]{},{}\<[21E]%
\>[24]{}\{\mskip1.5mu {\Varid{a}^{\Varid{x}}}\cdot{\Varid{a}^{\Varid{y}}}\mathrel{=}{\Varid{a}^{\Varid{x}\cdot\Varid{y}}}\;\quad\;{}\<[76]%
\>[76]{}[\mskip1.5mu \textsc{BugAddExp}\mskip1.5mu]\mskip1.5mu\}{}\<[E]%
\\
\>[21]{},{}\<[21E]%
\>[24]{}\Varid{eqPower},{}\<[34]%
\>[34]{}\Varid{simPower},\Varid{suitablePower},{}\<[60]%
\>[60]{}\Varid{readyPower}{}\<[E]%
\\
\>[21]{},{}\<[21E]%
\>[24]{}\textsc{DistExp}\mathbin{<}\textsc{MulExp}\mathbin{<}\textsc{AddExp}){}\<[E]%
\ColumnHook
\end{hscode}\resethooks
Here, \ensuremath{\Varid{eqPower}}, \ensuremath{\Varid{simPower}}, \ensuremath{\Varid{suitablePower}}, and \ensuremath{\Varid{readyPower}} are defined as
follows, where \ensuremath{\simeq} stands for syntactic equality:
\begin{hscode}\SaveRestoreHook
\column{B}{@{}>{\hspre}l<{\hspost}@{}}%
\column{3}{@{}>{\hspre}l<{\hspost}@{}}%
\column{20}{@{}>{\hspre}l<{\hspost}@{}}%
\column{40}{@{}>{\hspre}l<{\hspost}@{}}%
\column{46}{@{}>{\hspre}l<{\hspost}@{}}%
\column{E}{@{}>{\hspre}l<{\hspost}@{}}%
\>[3]{}\Varid{normPower}\;\Varid{v}{}\<[40]%
\>[40]{}\mathrel{=}\Varid{v}{}\<[E]%
\\
\>[3]{}\Varid{normPower}\;(\Varid{e}_{1}\cdot\Varid{e}_{2}){}\<[40]%
\>[40]{}\mathrel{=}\Varid{simplifyPower}\;(\Varid{normPower}\;\Varid{e}_{1}\cdot\Varid{normPower}\;\Varid{e}_{2}){}\<[E]%
\\
\>[3]{}\Varid{normPower}\;({\Varid{e}^{\Varid{n}}}){}\<[40]%
\>[40]{}\mathrel{=}\Varid{simplifyPower}\;({(\Varid{normPower}\;\Varid{e})^{\Varid{n}}}){}\<[E]%
\\[\blanklineskip]%
\>[3]{}\Varid{simplifyPower}\;({({\Varid{a}^{\Varid{x}}})^{\Varid{y}}}){}\<[46]%
\>[46]{}\mathrel{=}{\Varid{a}^{\Varid{x}\cdot\Varid{y}}}{}\<[E]%
\\
\>[3]{}\Varid{simplifyPower}\;({\Varid{a}^{\Varid{x}}}\cdot{\Varid{a}^{\Varid{y}}}){}\<[46]%
\>[46]{}\mathrel{=}{\Varid{a}^{\Varid{x}\mathbin{+}\Varid{y}}}{}\<[E]%
\\
\>[3]{}\Varid{simplifyPower}\;({(\Varid{a}\cdot\Varid{b})^{\Varid{x}}}){}\<[46]%
\>[46]{}\mathrel{=}{\Varid{a}^{\Varid{x}}}\cdot{\Varid{b}^{\Varid{x}}}{}\<[E]%
\\
\>[3]{}\Varid{simplifyPower}\;\Varid{e}{}\<[46]%
\>[46]{}\mathrel{=}\Varid{e}{}\<[E]%
\\[\blanklineskip]%
\>[3]{}\Varid{eqPower}\;\Varid{e}_{1}\;\Varid{e}_{2}{}\<[20]%
\>[20]{}\mathrel{=}\Varid{normPower}\;\Varid{e}_{1}\simeq\Varid{normPower}\;\Varid{e}_{2}{}\<[E]%
\\[\blanklineskip]%
\>[3]{}\Varid{simPower}\;\Varid{e}_{1}\;\Varid{e}_{2}{}\<[20]%
\>[20]{}\mathrel{=}\Varid{e}_{1}\simeq\Varid{e}_{2}{}\<[E]%
\\[\blanklineskip]%
\>[3]{}\Varid{suitablePower}\;\Varid{e}{}\<[20]%
\>[20]{}\mathrel{=}\Varid{normPower}\;\Varid{e}\;\not\simeq\;\Varid{e}{}\<[E]%
\\[\blanklineskip]%
\>[3]{}\Varid{readyPower}\;\Varid{e}{}\<[20]%
\>[20]{}\mathrel{=}\Varid{normPower}\;\Varid{e}\simeq\Varid{e}{}\<[E]%
\ColumnHook
\end{hscode}\resethooks
We have simplified the implementation of \ensuremath{\Varid{normPower}} by ignoring the fact that
\ensuremath{\cdot} is associative. 
In this example, we have chosen to use normalisation as a
base to define the exercise-specific functions, such as \ensuremath{\equiv }. However, this is
not necessarily the case for other domains.

\subsection{Formalised Services}
\label{ssec:formalised-services}
We define the set of services we offer to learning environments in terms of the
relations defined in the previous section.

\begin{description}
\item[allfirsts.] 
The \ensuremath{\Varid{allfirsts}} service returns all next steps that are allowed by a strategy in 
a particular state:
\begin{hscode}\SaveRestoreHook
\column{B}{@{}>{\hspre}l<{\hspost}@{}}%
\column{3}{@{}>{\hspre}l<{\hspost}@{}}%
\column{E}{@{}>{\hspre}l<{\hspost}@{}}%
\>[3]{}\Varid{allfirsts}\;\Conid{S}_{0}\mathrel{=}\{\mskip1.5mu (\Conid{S},\Varid{r})\mid \Conid{S}_{0}\;{\buildrel \Varid{r} \over \twoheadrightarrow}\;\Conid{S}\mskip1.5mu\}{}\<[E]%
\ColumnHook
\end{hscode}\resethooks
Consider the following state \ensuremath{\Conid{S}}:
\begin{hscode}\SaveRestoreHook
\column{B}{@{}>{\hspre}l<{\hspost}@{}}%
\column{3}{@{}>{\hspre}l<{\hspost}@{}}%
\column{59}{@{}>{\hspre}l<{\hspost}@{}}%
\column{E}{@{}>{\hspre}l<{\hspost}@{}}%
\>[3]{}\Conid{S}\mathrel{=}(\Gamma,{\llbracket{({\Varid{a}^{\mathrm{3}}}\cdot{\Varid{a}^{\mathrm{4}}})^{\mathrm{2}}}\rrbracket},{}\<[59]%
\>[59]{}((\Varid{somewhere}\;\textsc{AddExp})\mathrel{{<}\hspace{-0.4em}\star\hspace{-0.4em}{>}}\textsc{MulExp})\mathrel{{<}\hspace{-0.4em}\mid\hspace{-0.4em}{>}}{}\<[E]%
\\
\>[59]{}(\textsc{DistExp}\mathrel{{<}\hspace{-0.4em}\star\hspace{-0.4em}{>}}\Varid{repeat}\;\textsc{MulExp}\mathrel{{<}\hspace{-0.4em}\star\hspace{-0.4em}{>}}\textsc{AddExp})){}\<[E]%
\ColumnHook
\end{hscode}\resethooks
An \ensuremath{\Varid{allfirsts}\;\Conid{S}} service call gives the following result. The \ensuremath{\textsc{Up}} minor rule is
introduced by the somewhere combinator to get the focus back to its original
place.
\begin{hscode}\SaveRestoreHook
\column{B}{@{}>{\hspre}l<{\hspost}@{}}%
\column{4}{@{}>{\hspre}c<{\hspost}@{}}%
\column{4E}{@{}l@{}}%
\column{7}{@{}>{\hspre}l<{\hspost}@{}}%
\column{74}{@{}>{\hspre}l<{\hspost}@{}}%
\column{102}{@{}>{\hspre}l<{\hspost}@{}}%
\column{103}{@{}>{\hspre}l<{\hspost}@{}}%
\column{111}{@{}>{\hspre}l<{\hspost}@{}}%
\column{112}{@{}>{\hspre}c<{\hspost}@{}}%
\column{112E}{@{}l@{}}%
\column{E}{@{}>{\hspre}l<{\hspost}@{}}%
\>[4]{}\{\mskip1.5mu {}\<[4E]%
\>[7]{}((\Gamma,{({\llbracket{\Varid{a}^{\mathrm{7}}}\rrbracket})^{\mathrm{2}}},{}\<[74]%
\>[74]{}\textsc{Up}\mathrel{{<}\hspace{-0.4em}\star\hspace{-0.4em}{>}}\textsc{MulExp}),{}\<[103]%
\>[103]{}\textsc{AddExp}{}\<[112]%
\>[112]{}){}\<[112E]%
\\
\>[4]{},{}\<[4E]%
\>[7]{}((\Gamma,{\llbracket{({\Varid{a}^{\mathrm{3}}})^{\mathrm{2}}}\cdot{({\Varid{a}^{\mathrm{4}}})^{\mathrm{2}}}\rrbracket},{}\<[74]%
\>[74]{}\Varid{repeat}\;\textsc{MulExp}\mathrel{{<}\hspace{-0.4em}\star\hspace{-0.4em}{>}}\textsc{AddExp}),{}\<[102]%
\>[102]{}\textsc{DistExp}{}\<[111]%
\>[111]{})\mskip1.5mu\}{}\<[E]%
\ColumnHook
\end{hscode}\resethooks

\item[onefirst.] 
The \ensuremath{\Varid{onefirst}} service returns a single possible next step that follows the 
strategy. This service uses the \ensuremath{\Varid{allfirsts}} service and the rule ordering.
\begin{hscode}\SaveRestoreHook
\column{B}{@{}>{\hspre}l<{\hspost}@{}}%
\column{3}{@{}>{\hspre}l<{\hspost}@{}}%
\column{31}{@{}>{\hspre}l<{\hspost}@{}}%
\column{E}{@{}>{\hspre}l<{\hspost}@{}}%
\>[3]{}\Varid{onefirst}\;\Conid{S}_{0}\mathrel{=}(\Conid{S},\Varid{r})\;\mathbf{where}\;{}\<[31]%
\>[31]{}(\Conid{S},\Varid{r})\;\mathrel{\in}\;\Varid{allfirsts}\;\Conid{S}_{0}\;\wedge\;\forall\;(\mathit{S_i},\mathit{r_i})\;\mathrel{\in}\;\Varid{allfirsts}\;\Conid{S}_{0}\;.\;\Varid{r}\;\leq\;\mathit{r_i}{}\<[E]%
\ColumnHook
\end{hscode}\resethooks
Performing a \ensuremath{\Varid{onefirst}} service call on the state \ensuremath{\Conid{S}} from the previous
paragraph, taking into account the rule ordering from
subsection~\ref{ssec:non-determinism}, gives following result:
\begin{hscode}\SaveRestoreHook
\column{B}{@{}>{\hspre}l<{\hspost}@{}}%
\column{3}{@{}>{\hspre}l<{\hspost}@{}}%
\column{E}{@{}>{\hspre}l<{\hspost}@{}}%
\>[3]{}((\Gamma,{({\llbracket{\Varid{a}^{\mathrm{7}}}\rrbracket})^{\mathrm{2}}},\textsc{Up}\mathrel{{<}\hspace{-0.4em}\star\hspace{-0.4em}{>}}\textsc{MulExp}),\textsc{AddExp}){}\<[E]%
\ColumnHook
\end{hscode}\resethooks
\item[derivation.]
The \ensuremath{\Varid{derivation}} service returns a worked-out solution of an exercise starting
with the current expression.
\begin{hscode}\SaveRestoreHook
\column{B}{@{}>{\hspre}l<{\hspost}@{}}%
\column{3}{@{}>{\hspre}l<{\hspost}@{}}%
\column{20}{@{}>{\hspre}l<{\hspost}@{}}%
\column{E}{@{}>{\hspre}l<{\hspost}@{}}%
\>[3]{}\Varid{derivation}\;\Conid{S}_{0}\mathrel{=}{}\<[20]%
\>[20]{}(\Conid{S}_{1},\Varid{r}_{1})\;(\Conid{S}_{2},\Varid{r}_{2})\;\ldots\;(\mathit{S_n},\mathit{r_n})\;\mathbf{where}\;\Varid{empty}\;(\mathit{S_n})\;\wedge{}\<[E]%
\\
\>[20]{}\forall\;\Varid{i}\;\mathrel{\in}\;\mathrm{1}\;\ldots\;\Varid{n}\;.\;(\mathit{S_{i-1}},\mathit{r_i})\mathrel{=}\Varid{onefirst}\;\mathit{S_i}{}\<[E]%
\ColumnHook
\end{hscode}\resethooks
\item[ready.]
The \ensuremath{\Varid{ready}} service checks if the expression in a state is in a form accepted as
a final answer. The \ensuremath{\Varid{ready}} service is an interface to the \ensuremath{\Varid{isReady}} predicate
defined in an exercise.

\item[stepsremaining.]
The \ensuremath{\Varid{stepsremaining}} service computes how many steps remain to be done according
to the strategy. This is achieved by calculating the length of a derivation.

\item[apply.]
The \ensuremath{\Varid{apply}} service applies a rule to an expression at a particular location,
regardless of the strategy. The location is represented as a list of integers,
where each integer \ensuremath{\Varid{n}} represents the number of steps to the right after a step
downwards in a subexpression (the \ensuremath{\Varid{n}}th child). Starting at the root of an
expression we can assign every subexpression a unique location.

The function \ensuremath{\Varid{setFocus}} translates a location to a sequence of minor rules that
puts the focus at a particular subexpression.
\begin{hscode}\SaveRestoreHook
\column{B}{@{}>{\hspre}l<{\hspost}@{}}%
\column{3}{@{}>{\hspre}l<{\hspost}@{}}%
\column{25}{@{}>{\hspre}l<{\hspost}@{}}%
\column{E}{@{}>{\hspre}l<{\hspost}@{}}%
\>[3]{}\Varid{setFocus}\;\Conid{S}_{0}\;[\mskip1.5mu \mskip1.5mu]{}\<[25]%
\>[25]{}\mathrel{=}\Conid{S}_{0}{}\<[E]%
\\
\>[3]{}\Varid{setFocus}\;\Conid{S}_{0}\;(\Varid{n}\mathbin{:}\Varid{ns}){}\<[25]%
\>[25]{}\mathrel{=}\Varid{setFocus}\;\Conid{S}_{1}\;\Varid{ns}\;\mathbf{where}\;\Conid{S}_{0}\;\stackrel{(\textsc{Down}\;(\Varid{const}\;\Varid{n}))}{\rightarrow}\;\Conid{S}_{1}{}\<[E]%
\ColumnHook
\end{hscode}\resethooks
The function \ensuremath{\Varid{focusToRoot}} sets the focus to the root of an expression. We omit
the definition of this function. The \ensuremath{\Varid{apply}} service is defined as follows:
\begin{hscode}\SaveRestoreHook
\column{B}{@{}>{\hspre}l<{\hspost}@{}}%
\column{3}{@{}>{\hspre}l<{\hspost}@{}}%
\column{30}{@{}>{\hspre}l<{\hspost}@{}}%
\column{E}{@{}>{\hspre}l<{\hspost}@{}}%
\>[3]{}\Varid{apply}\;\Varid{r}\;\Varid{loc}\;\Conid{S}_{0}\mathrel{=}\Conid{S}_{1}\;\mathbf{where}\;{}\<[30]%
\>[30]{}\Varid{setFocus}\;\Varid{loc}\;(\Varid{focusToRoot}\;\Conid{S}_{0})\;{\buildrel \Varid{r} \over \rightarrow}\;\Conid{S}_{1}{}\<[E]%
\ColumnHook
\end{hscode}\resethooks
For example, \ensuremath{\Varid{apply}\;\textsc{MulExp}\;[\mskip1.5mu \mathrm{1}\mskip1.5mu]\;(\Gamma,{\llbracket{({\Varid{a}^{\mathrm{3}}})^{\mathrm{2}}}\cdot{({\Varid{a}^{\mathrm{4}}})^{\mathrm{2}}}\rrbracket},\sigma)} gives \ensuremath{(\Gamma,({({\Varid{a}^{\mathrm{3}}})^{\mathrm{2}}}\cdot{\llbracket{\Varid{a}^{\mathrm{8}}}\rrbracket}),\sigma)}.

\item[applicable.]
The \ensuremath{\Varid{applicable}} service takes an expression and a location in this expression,
and returns all major rules that can be applied at this location, independent of
the strategy. Let \ensuremath{\Conid{R}} be the union of the rules in the strategy and the exercise 
rule set, then \ensuremath{\Varid{applicable}} is defined as follows:
\begin{hscode}\SaveRestoreHook
\column{B}{@{}>{\hspre}l<{\hspost}@{}}%
\column{3}{@{}>{\hspre}l<{\hspost}@{}}%
\column{E}{@{}>{\hspre}l<{\hspost}@{}}%
\>[3]{}\Varid{applicable}\;\Varid{loc}\;\Conid{S}_{0}\mathrel{=}\{\mskip1.5mu \Varid{r}\mid \Varid{r}\;\mathrel{\in}\;\Conid{R},\Conid{S}_{1}\;{\buildrel \Varid{r} \over \rightarrow}\;\Conid{S}_{2}\mskip1.5mu\}\;\mathbf{where}\;\Conid{S}_{1}\mathrel{=}\Varid{setFocus}\;\Varid{loc}\;(\Varid{focusToRoot}\;\Conid{S}_{0}){}\<[E]%
\ColumnHook
\end{hscode}\resethooks
\item[generate.]
The \ensuremath{\Varid{generate}} service takes an exercise code and a difficulty level (optional),
and returns an initial state with a freshly generated expression.

\item[diagnose.]
The \ensuremath{\Varid{diagnose}} service diagnoses an expression submitted by a student. Possible
diagnoses are:
\begin{itemize}
\item \ensuremath{\Conid{NotEq}}: the current and submitted expression are not equivalent, i.e.
  something is wrong,
\item \ensuremath{\Conid{Buggy}}: a common misconception has been detected,
\item \ensuremath{\Conid{Similar}}: the expression is similar to the last expression in the
  derivation,
\item \ensuremath{\Conid{Expected}}: the submitted expression is expected by the strategy,
\item \ensuremath{\Conid{Detour}}: the submitted expression was not expected by the strategy, but
  the applied rule was detected,
\item \ensuremath{\Conid{Correct}}: the submitted expression is correct, but we cannot determine
  which rule was applied.
\end{itemize}
\noindent
The \ensuremath{\Varid{diagnose}} service is defined as follows, where the \ensuremath{\Varid{unFocus}} function
converts a focussed expression to a normal expression:
\begin{hscode}\SaveRestoreHook
\column{B}{@{}>{\hspre}l<{\hspost}@{}}%
\column{3}{@{}>{\hspre}l<{\hspost}@{}}%
\column{5}{@{}>{\hspre}l<{\hspost}@{}}%
\column{8}{@{}>{\hspre}l<{\hspost}@{}}%
\column{94}{@{}>{\hspre}l<{\hspost}@{}}%
\column{100}{@{}>{\hspre}l<{\hspost}@{}}%
\column{110}{@{}>{\hspre}l<{\hspost}@{}}%
\column{E}{@{}>{\hspre}l<{\hspost}@{}}%
\>[3]{}\Varid{diagnose}\;(\Gamma,\varphi ,\sigma)\;\Varid{new}\mathrel{=}{}\<[E]%
\\
\>[3]{}\hsindent{2}{}\<[5]%
\>[5]{}\mathbf{if}\;\Varid{unFocus}\;\varphi \;\not\equiv\;\Varid{new}\;{}\<[94]%
\>[94]{}\mathbf{then}{}\<[E]%
\\
\>[5]{}\hsindent{3}{}\<[8]%
\>[8]{}\mathbf{if}\;\exists\;\Varid{b}\;\mathrel{\in}\;\Varid{buggyRuleSet}\;.\;\exists\;\Varid{e}\;.\;\Varid{unFocus}\;\varphi \;\stackrel{\Varid{b}}{\leadsto}\;\Varid{e}\;\wedge\;\Varid{e}\equiv \Varid{new}\;{}\<[94]%
\>[94]{}\mathbf{then}\;{}\<[100]%
\>[100]{}\Conid{Buggy}{}\<[E]%
\\
\>[94]{}\mathbf{else}\;{}\<[100]%
\>[100]{}\Conid{NotEq}\;{}\<[110]%
\>[110]{}\mathbf{else}{}\<[E]%
\\
\>[3]{}\hsindent{2}{}\<[5]%
\>[5]{}\mathbf{if}\;\Varid{unFocus}\;\varphi \;\approx\;\Varid{new}\;{}\<[94]%
\>[94]{}\mathbf{then}\;{}\<[100]%
\>[100]{}\Conid{Similar}\;{}\<[110]%
\>[110]{}\mathbf{else}{}\<[E]%
\\
\>[3]{}\hsindent{2}{}\<[5]%
\>[5]{}\mathbf{if}\;\Varid{new}\;\mathrel{\in}\;\Varid{allfirsts}\;(\Gamma,\varphi ,\sigma)\;{}\<[94]%
\>[94]{}\mathbf{then}\;{}\<[100]%
\>[100]{}\Conid{Expected}\;{}\<[110]%
\>[110]{}\mathbf{else}{}\<[E]%
\\
\>[3]{}\hsindent{2}{}\<[5]%
\>[5]{}\mathbf{if}\;\exists\;\Varid{r}\;\mathrel{\in}\;\Varid{ruleSet}\;.\;\exists\;\Varid{e}\;.\;\Varid{unFocus}\;\varphi \;{\buildrel \Varid{r} \over \leadsto}\;\Varid{e}\;\wedge\;\Varid{e}\equiv \Varid{new}\;{}\<[94]%
\>[94]{}\mathbf{then}\;{}\<[100]%
\>[100]{}\Conid{Detour}{}\<[E]%
\\
\>[94]{}\mathbf{else}\;{}\<[100]%
\>[100]{}\Conid{Correct}{}\<[E]%
\ColumnHook
\end{hscode}\resethooks
\end{description}

We conclude this section with a soundness result. We give a theorem connecting
the language concept, the function \ensuremath{\Conid{L}}, to the services built on top of
strategies defined in this section. We start with the introduction of a lemma
that simplifies the proof of the theorem.

\begin{lemma}
\label{lem:split}
  A major language \ensuremath{\mathcal{L}} is the language of a strategy without occurrences of 
  minor rules:
\begin{hscode}\SaveRestoreHook
\column{B}{@{}>{\hspre}l<{\hspost}@{}}%
\column{3}{@{}>{\hspre}l<{\hspost}@{}}%
\column{E}{@{}>{\hspre}l<{\hspost}@{}}%
\>[3]{}\mathcal{L}\;(\sigma)\mathrel{=}\{\mskip1.5mu [\mskip1.5mu \Varid{a}\mid \Varid{a}\leftarrow \Varid{as},\neg \;\Varid{isMinor}\;\Varid{a}\mskip1.5mu]\mid \Varid{as}\;\mathrel{\in}\;\Conid{L}\;(\sigma)\mskip1.5mu\}{}\<[E]%
\ColumnHook
\end{hscode}\resethooks
  Let F be the set of all splits given by the split relation for \ensuremath{\sigma}:
  \ensuremath{\{\mskip1.5mu \sigma_i\mid \sigma\mapsto\sigma_i\mskip1.5mu\}}. Then the major language of \ensuremath{\sigma} without the empty sentence is
  equal to the union of major languages of the elements in F:
\begin{hscode}\SaveRestoreHook
\column{B}{@{}>{\hspre}l<{\hspost}@{}}%
\column{3}{@{}>{\hspre}l<{\hspost}@{}}%
\column{E}{@{}>{\hspre}l<{\hspost}@{}}%
\>[3]{}\mathcal{L}\;(\sigma)\;\setminus\;\{\mskip1.5mu \epsilon \mskip1.5mu\}\equiv \bigcup\limits_{\sigma_f \in F}\;\mathcal{L}\;(\sigma_f){}\<[E]%
\ColumnHook
\end{hscode}\resethooks
\end{lemma}
\begin{proof}
By case analysis on the structure of \ensuremath{\sigma}.
\end{proof}

\begin{thm}
  Let \ensuremath{\Conid{S}_{0}} be a state \ensuremath{(\Gamma\mathrel{\times}\varphi \mathrel{\times}\sigma)} and \ensuremath{\vec{\Varid{r}}} be a sequence of rules
  \ensuremath{[\mskip1.5mu \Varid{r}_{1}\;\ldots\;\mathit{r_n}\mskip1.5mu]} such that \ensuremath{\Varid{derivation}\;(\Conid{S}_{0})\mathrel{=}(\Conid{S}_{1},\Varid{r}_{1})\;\ldots\;(\mathit{S_n},\mathit{r_n})}. Then
  \ensuremath{\vec{\Varid{r}}\;\mathrel{\in}\;\mathcal{L}\;(\sigma)}.
\end{thm}
\begin{proof}
  We sketch a proof. Via the definitions of \ensuremath{\Varid{derivation}}, \ensuremath{\Varid{onefirst}}, and
  \ensuremath{\Varid{allfirsts}} we know that for every element in the derivation sequence \ensuremath{[\mskip1.5mu \Varid{r}_{1}\;\ldots\;\mathit{r_n}\mskip1.5mu]} a big step has to be performed.
  We distinguish two cases in our hypothesis \ensuremath{[\mskip1.5mu \Varid{r}_{1}\;\ldots\;\mathit{r_n}\mskip1.5mu]\;\mathrel{\in}\;\mathcal{L}\;(\sigma)}:
  \begin{description}
  \item [\ensuremath{\Varid{n}\mathrel{=}\mathrm{0}}]: From the statement above we deduce that the strategy \ensuremath{\sigma} is
    either \ensuremath{\gamma} or consists of a strategy with minor rules only. From the
    definition of the major language \ensuremath{\mathcal{L}} we know \ensuremath{\epsilon \;\mathrel{\in}\;\mathcal{L}\;(\sigma)}.
  \item [\ensuremath{\Varid{n}\mathbin{>}\mathrm{0}}]: This case is proven by induction on the length (\ensuremath{\Varid{n}}) of 
    the sequence \ensuremath{[\mskip1.5mu \Varid{r}_{1}\;\ldots\;\mathit{r_n}\mskip1.5mu]} using Lemma~\ref{lem:split}.
  \end{description}
\end{proof}

\section{Exercise Properties}
\label{sec:properties}
The following lemmas express properties that the components in an exercise
should have. The lemmas connect the various components of an exercise. The proof
of these lemmas for a particular exercise indicates the soundness of the
corresponding exercise specification.

\begin{lemma}
  Let \ensuremath{\Conid{S}_{0}} be a state \ensuremath{(\Gamma,{\llbracket\Varid{e}\rrbracket},\sigma)} and \ensuremath{\Varid{derivation}\;\Conid{S}_{0}\mathrel{=}(\Conid{S}_{1},\Varid{r}_{1})\;\ldots\;(\mathit{S_n},\mathit{r_n})}. If \ensuremath{\Varid{e}} is a suitable start expression, the expression in the last
  state \ensuremath{(\mathit{S_n})} is ready.
\end{lemma}
\begin{lemma}
  Let \ensuremath{\Varid{r}} be a rewrite rule for a particular exercise: \ensuremath{\Varid{e}_{1}\;{\buildrel \Varid{r} \over \leadsto}\;\Varid{e}_{2}}. Then \ensuremath{\Varid{r}} 
  is semantics preserving, i.e., \ensuremath{\Varid{e}_{1}\equiv \Varid{e}_{2}}.
\end{lemma}
\begin{corollary}
  Let \ensuremath{\Conid{S}_{0}} be a state and \ensuremath{\Varid{derivation}\;\Conid{S}_{0}\mathrel{=}(\Conid{S}_{1},\Varid{r}_{1})\;\ldots\;(\mathit{S_n},\mathit{r_n})}.  Then all
  expressions in \ensuremath{\Conid{S}_{0}} up to \ensuremath{\mathit{S_n}} are in the same equivalence class determined by
  the exercise's equivalence relation \ensuremath{\equiv }.
\end{corollary}
\begin{lemma}
  Let \ensuremath{\Conid{E}} be the set of expressions generated by the exercise's generator. 
  All expressions \ensuremath{\Varid{e}\;\mathrel{\in}\;\Conid{E}} are suitable and not ready: \ensuremath{\Varid{suitable}\;(\Varid{e})\;\wedge\;\neg\;\Varid{ready}\;(\Varid{e})}.
\end{lemma}
\begin{lemma}
  Let \ensuremath{\Conid{S}_{0}} be a state, \ensuremath{\Varid{onefirst}\;\Conid{S}_{0}\mathrel{=}(\Conid{S}_{1},\Varid{r}_{1})} and \ensuremath{\Varid{allfirsts}\;\Conid{S}_{0}\mathrel{=}\{\mskip1.5mu (\Conid{S}_{2},\Varid{r}_{2})\;\ldots\;(\mathit{S_n},\mathit{r_n})\mskip1.5mu\}}. Then \ensuremath{\Varid{r}_{1}\equiv \Varid{min}\;\{\mskip1.5mu \Varid{r}_{2}\;\ldots\;\mathit{r_n}\mskip1.5mu\}} with respect to the rule
  ordering.
\end{lemma}
The proofs of these lemmas for the exercise for calculating with powers, are
rather easy. This is because the rules that are used in the strategy are the
same as the rules used to calculate the normal form of a power expression. We
provide the proofs for these lemmas in an accompanying technical
report~\cite{UUCS2010028}.

\section{Related Work}
\label{sec:related}
The strategy language on which our work is based is very similar to languages
that are used in parser libraries~\cite{SwieDupo96}. Some differences, such as
the usage of labels, the focus on intermediate answers, and the concept of minor
rules, make the existing libraries unsuitable for generating feedback.  

Our language is also similar to strategic programming languages such as
Stratego~\cite{VBT98} and Elan~\cite{borovansky01rewriting}, and tactic
languages used in theorem proving. We compare our approach with two more formal
approaches to strategies. Kaiser and L\"ammel construct a mechanised formal
model of Stratego~\cite{kaiser-09}. Our strategy language is different from
Stratego in the sense that we, in addition to the final term, also focus on the
intermediate rewrite steps. We do not focus on the development of a mechanised
model: instead, we give a formal definition of our feedback services that use
our strategy language. Tacticals, proof plans and methods~\cite{bundyproofplans}
are used to automatically prove theorems. On an abstract level, these plans and
methods play the same role as strategies: we can view a strategy as a proof plan
for proving that an answer is the solution to an exercise. Aspinall et
al.~\cite{aspinall} introduce the tactic language Hitac that can be used to
construct hierarchical proofs, so called hiproofs. To evaluate Hitac programs
two semantics are given: a big step semantic that captures the intended meaning
and a small step semantic that covers the details of the proof. As far as we
found, the tactic language is not used to generate feedback, or to recognise
proving steps made by a student. Moreover, we provide a set of services that
enables learning environments to access our functionality.

\section{Conclusions}
\label{sec:conclusions}
In this paper we have presented a formal and precise definition of the main
concepts that we use to construct semantically rich feedback in learning
environments. In addition to a precise definition of our strategy language we
define a step, and a big step relation. These relations give the semantics of
the strategy language. Feedback services are an interface to our feedback
functionality that can be used by learning environments. These services are
expressed in terms of the big step relation.

Our formalisation gives us more confidence in the design choices we have made.
Furthermore, we can now validate our current implementation using the properties
we state. In the future we want to extend the work in this paper by providing
proofs for other domains, such as linear algebra and propositional logic.

\providecommand{\bibfont}{\small}

\label{sec:bib}
\bibliographystyle{eptcs}
\bibliography{StrategyProperties}

\end{document}